\documentclass{llncs}

\usepackage[dvipdfmx]{graphicx}
\usepackage{amssymb}
\usepackage{amsmath}
\usepackage{enumerate}
\title{Parameterized Complexity of the List Coloring Reconfiguration Problem with Graph Parameters\thanks{This work is partially supported by JST CREST Grant Number
	JPMJCR1402, and by JSPS KAKENHI Grant Numbers JP16J02175, JP16K00003, and JP16K00004.}}

\author{
	Tatsuhiko Hatanaka
\and
	Takehiro Ito
\and
	Xiao Zhou
}

\institute{
	Graduate School of Information Sciences, 
	Tohoku University, \\
	Aoba-yama 6-6-05, Sendai 980-8579, Japan.\\
	\email{\{hatanaka, takehiro, zhou\}@ecei.tohoku.ac.jp}
}

\newcommand{\Neigh}[2]{N(#1,#2)}
\newcommand{\mxcli}[1]{\omega (#1)}
\newcommand{\NN}{\mathbb{N}}
\newcommand{\VC}{\mathsf{vc}}

\newcommand{\lcr}{\textsc{list coloring reconfiguration}}
\newcommand{\kcr}{\textsc{coloring reconfiguration}}
\newcommand{\lcsr}{\textsc{list coloring shortest reconfiguration}}

\newcommand{\mapf}{f}
\newcommand{\mapg}{g}
\newcommand{\verlist}{L}
\newcommand{\colorset}{C}
\newcommand{\numk}{k}
\newcommand{\cind}[1]{c_{#1}}

\newcommand{\ini}{0}
\newcommand{\tar}{t}
\newcommand{\seq}[1]{\mathcal{#1}}
\newcommand{\Seq}{\seq{S}}

\newcommand{\diff}[2]{\mathsf{dif}(#1,#2)}

\newcommand{\Inst}{\seq{I}}

\newcommand{\weight}{w}
\newcommand{\gap}[3]{\mathsf{gap}_{#1}(#2,#3)}
\newcommand{\len}[2]{\mathsf{len}_{#1}(#2)}
\newcommand{\OPT}[2]{\mathsf{OPT}(#1,#2)}
\newcommand{\nrec}[2]{\# (#1,#2)}
 
\newcommand{\MD}[1]{\mathsf{MD}(#1)}
\newcommand{\PMD}[1]{\mathsf{PMD}(#1)}
\newcommand{\mw}[1]{\mathsf{mw}(#1)}
\newcommand{\pmw}[1]{\mathsf{pmw}(#1)}
\newcommand{\MW}{\mathsf{mw}}
\newcommand{\PMW}{\mathsf{pmw}}

\newcommand{\SUB}[2]{\mathsf{Sub}(#1,#2)}
\newcommand{\Quo}[1]{\mathsf{Q}(#1)}
\newcommand{\coG}[1]{\mathsf{CG}(#1)}

\newcommand{\asgn}[1]{A(#1)}

\newcommand{\kernel}[1]{h_{k,\PMW}(#1)}

\newcommand{\IDM}[1]{\mathcal{M}_m(#1)}
\newcommand{\IDMi}[2]{\mathcal{M}_{#2}(#1)}

\newcommand{\Vcov}{V_C}
\newcommand{\Vind}{V_I}

\newcommand{\hset}[1]{\{ #1\}}
\newcommand{\hseq}[1]{\langle #1\rangle}

\newcommand{\rest}[2]{{#1^{#2}}}

\newcommand{\key}{c^*}
\newcommand{\select}{V_\mathsf{sel}}

\newcommand{\Vfbd}{V_\mathsf{for}}
\newcommand{\fbd}[2]{w}

\newcounter{one}
\newcounter{two}
\newcounter{three}
\newcounter{four}
\newcounter{five}
\newcounter{six}
\newtheorem{observation}{Observation}

\begin{document}
\maketitle
\sloppy

\begin{abstract}
	Let $G$ be a graph such that each vertex has its list of available colors, and assume that each list is a subset of the common set consisting of $k$ colors.
	For two given list colorings of $G$, we study the problem of transforming one into the other by changing only one vertex color assignment at a time, while at all times maintaining a list coloring.
	This problem is known to be PSPACE-complete even for bounded bandwidth graphs and a fixed constant $k$.
	In this paper, we study the fixed-parameter tractability of the problem when parameterized by several graph parameters.
	We first give a fixed-parameter algorithm for the problem when parameterized by $k$ and the modular-width of an input graph.
	We next give a fixed-parameter algorithm for the shortest variant when parameterized by $k$ and the size of a minimum vertex cover of an input graph.
	As corollaries, we show that the problem for cographs and the shortest variant for split graphs are fixed-parameter tractable even when only $k$ is taken as a parameter.
	On the other hand, we prove that the problem is W[1]-hard when parameterized only by the size of a minimum vertex cover of an input graph.
\end{abstract}

\section{Introduction}

	Recently, the framework of \emph{reconfiguration}~\cite{IDHPSUU} has been extensively studied in the field of theoretical computer science.
	This framework models several situations where we wish to find a step-by-step transformation between two feasible solutions of a combinatorial (search) problem such that all intermediate solutions are also feasible and each step respects a fixed reconfiguration rule.
	This reconfiguration framework has been applied to several well-studied combinatorial problems.
	(See a survey~\cite{Jan13}.)
	
	\subsection{Our problem}
	
	In this paper, we study a reconfiguration problem for list (vertex) colorings in a graph, which was introduced by Bonsma and Cereceda~\cite{BC09}.
	
	Let $\colorset = \{ \cind{1}, \cind{2}, \ldots, \cind{\numk} \}$ be the set of $\numk$ colors, called the \emph{color set}.
	A (proper) $\numk$\emph{-coloring} of a graph $G = (V,E)$ is a mapping $\mapf \colon V \to \colorset$ such that $\mapf(v) \neq \mapf(w)$ for every edge $vw \in E$.
	In \emph{list coloring}, each vertex $v \in V$ has a set $\verlist(v) \subseteq \colorset$ of colors, called the \emph{list of} $v$;
	sometimes, the list assignment $\verlist\colon V\to 2^C$ itself is called a \emph{list}.
	Then, a $\numk$-coloring $\mapf$ of $G$ is called an \emph{$\verlist$-coloring} of $G$ if $\mapf(v) \in \verlist(v)$ holds for every vertex $v \in V$.
	Therefore, a $\numk$-coloring of $G$ is simply an $\verlist$-coloring of $G$ when $\verlist(v) = \colorset$ holds for every vertex $v$ of $G$, and hence $\verlist$-coloring is a generalization of $k$-coloring.
	Figure~\ref{fig:example}(b) illustrates four $\verlist$-colorings of the same graph $G$ in \figurename ~1(a); the color assigned to each vertex is attached to the vertex.
	
	In the reconfiguration framework, two $\verlist$-colorings $\mapf$ and $\mapf^\prime$ of a graph $G=(V,E)$ are said to be \emph{adjacent} if $|\hset{v \in V \colon \mapf(v) \neq \mapf^\prime(v)}|=1$ holds, that is, $\mapf^\prime$ can be obtained from $\mapf$ by recoloring exactly one vertex.
	A sequence $\hseq{\mapf_0, \mapf_1,\allowbreak \ldots, \mapf_\ell}$ of $\verlist$-colorings of $G$ is called a \emph{reconfiguration sequence} between $\mapf_\ini$ and $\mapf_\ell$ (of \emph{length} $\ell$) if $\mapf_{i-1}$ and $\mapf_i$ are adjacent for each $i \in \hset{1,2,\ldots, \ell}$.
	Two $\verlist$-colorings $\mapf$ and $\mapf^\prime$ are \emph{reconfigurable} if there exists a reconfiguration sequence between them.
	The {\sc list coloring reconfiguration} problem is to determine whether two given $\verlist$-colorings $\mapf_{\ini}$ and $\mapf_{\tar}$ are reconfigurable, or not. 
	Figure~\ref{fig:example} shows an example of a yes-instance of \lcr, where the vertex whose color assignment was changed from the previous one is depicted by a black circle.

	\begin{figure}[t]
		\begin{center}
			\includegraphics[width=0.85\textwidth,clip]{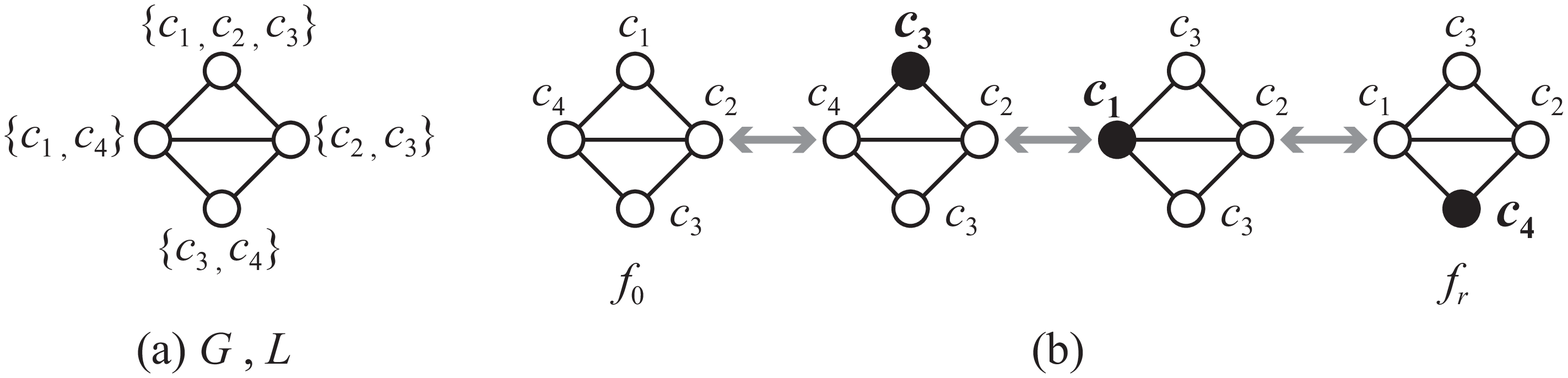}
		\end{center}
		\caption{A reconfiguration sequence between two $\verlist$-colorings $\mapf_\ini$ and $\mapf_\tar$ of $G$.}
		\label{fig:example}
	\end{figure}
	
	\subsection{Known and related results}
	
	\textsc{List coloring reconfiguration} is one of the most well-studied reconfiguration problems, as well as \kcr~which is a special case of the problem such that $\verlist(v) = \{ \cind{1}, \cind{2}, \ldots, \cind{\numk} \}$ holds for every vertex $v$.
	These problems have been studied intensively from various viewpoints~\cite{BB13,BC09,BJLPP14,BMNR14,Cer07,CHJ11,HIZ14,JKKPP14,Wro14} including the generalizations~\cite{BMMN16,Wro15}.
	
	Bonsma and Cereceda~\cite{BC09} proved that \kcr~is \allowbreak PSPACE-complete even for bipartite graphs and any fixed constant $k \ge 4$.
	On the other hand, Cereceda et al.~\cite{CHJ11} gave a polynomial-time algorithm solving \kcr ~for any graph and $k\le 3$; the algorithm can be applied to \lcr, too.
	In particular, the former result implies that there is no fixed-parameter algorithm for \kcr ~(and hence \lcr) when parameterized by only $k$ under the assumption of $\mathrm{P}\ne \mathrm{PSPACE}$.
	
	Bonsma et al.~\cite{BMNR14} and Johnson et al.~\cite{JKKPP14} independently developed a fixed-parameter algorithm to solve \kcr~when parameterized by $\numk + \ell$, where $\ell$ is the upper bound on the length of reconfiguration sequences, and again their algorithms can be applied to \lcr.
	In contrast, if \textsc{coloring reconfiguration} is parameterized only by $\ell$, then it is W[1]-hard when $\numk$ is an input~\cite{BMNR14} and does not admit a polynomial kernelization when $\numk$ is fixed unless the polynomial hierarchy collapses~\cite{JKKPP14}.
	
	Hatanaka et al.~\cite{HIZ14} proved that \lcr~is \allowbreak PSPACE-complete even for complete split graphs, whose modular-width is zero.
	Wrochna~\cite{Wro14} proved that \lcr~is PSPACE-complete even when $k$ and the bandwidth of an input graph are bounded by some constant; thus the treewidth and the cliquewidth of an input graph are also bounded.
	
	
	\begin{figure}[t]
		\begin{center}
			\includegraphics[width=\textwidth,clip]{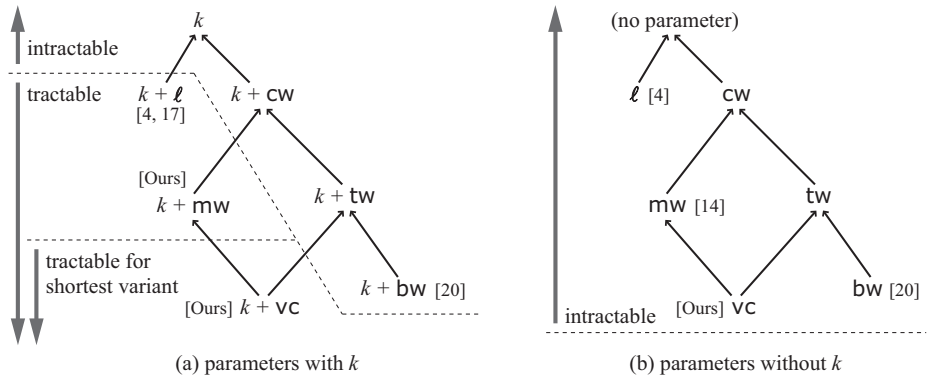}
		\end{center}
		\caption{All results (including known ones) for \textsc{list coloring reconfiguration} from the viewpoint of parameterized complexity, where $\mathsf{cw}$, $\mathsf{tw}$, $\mathsf{bw}$, $\mathsf{mw}$, and $\mathsf{vc}$ are the upper bounds on the cliquewidth, treewdith, bandwidth, modular-width, and the size of a minimum vertex cover of an input graph, respectively.}
		\label{fig:result}
	\end{figure}
	
	\subsection{Our contribution}
	
	To the best of our knowledge, known algorithmic results mostly employed the length $\ell$ of reconfiguration sequences as a parameter~\cite{BMNR14,JKKPP14}, and no fixed-parameter algorithm is known when parameterized by graph parameters. 
	Therefore, we study \textsc{list coloring reconfiguration} when parameterized by several graph parameters, and paint an interesting map of graph parameters which shows the boundary between fixed-parameter tractability and intractability.
	Our map is \figurename~\ref{fig:result} which shows both known and our results, where an arrow $\alpha \to \beta$ indicates that the parameter $\alpha$ is ``stronger'' than $\beta$, that is, $\beta$ is bounded if $\alpha$ is bounded.
	(For relationships of parameters, see, e.g.,~\cite{GLO13,KNR92}.)
	
	
	More specifically, we first give a fixed-parameter algorithm solving \lcr~when parameterized by $k$ and the modular-width $\MW$ of an input graph.
	(The definition of modular-width will be given in Section~\ref{sec:MD}.)
	Note that, according to the known results~\cite{BC09,HIZ14}, we cannot construct a fixed-parameter algorithm for general graphs when only one of $k$ and $\MW$ is taken as a parameter under the assumption of $\mathrm{P}\ne \mathrm{PSPACE}$.
	However, as later shown in Corollary~\ref{cor:cograph}, our algorithm implies that the problem is fixed-parameter tractable for cographs even when only $k$ is taken as a parameter.
	
	We then consider the shortest variant which computes the length of a shortest reconfiguration sequence (i.e., the minimum number of recoloring steps) for a yes-instance of \textsc{list coloring reconfiguration}, and show that it admits a fixed-parameter algorithm when parameterized by $k$ and the size of a minimum vertex cover of an input graph.
	Moreover, as a corollary, we show that the shortest variant is fixed-parameter tractable for split graphs even when only $k$ is taken as a parameter.
	
	Finally, we prove that \lcr~is $W[1]$-hard when parameterized only by the size of a minimum vertex cover of an input graph. 
	
	%
	
	\section{Preliminaries}
	\label{sec:pre}
	
	We assume without loss of generality that graphs are simple and connected.
	Let $G=(V,E)$ be a graph with vertex set $V$ and edge set $E$;
	we sometimes denote by $V(G)$ and $E(G)$ the vertex set and the edge set of $G$, respectively. 
	For a vertex $v$ in $G$, we denote by $\Neigh{G}{v}$ the neighborhood $\hset{u\in V\colon uv\in E}$ of $v$ in $G$. 
	For a vertex subset $V^\prime \subseteq V$, we denote by $G[V^\prime]$ the subgraph of $G$ induced by $V^\prime$, and denote $G\setminus V^\prime=G[V(G)\setminus V^\prime]$.
	For a subgraph $H$ of $G$, we denote $G\setminus H=G\setminus V(H)$.
	Let $\mxcli{G}$ be the size of a maximum clique of $G$.
	We have the following simple observation.
	\begin{observation}
		\label{obs:colors}
		Let $G$ be a graph with a list $\verlist \colon V(G)\to 2^\colorset$.
		If $G$ has an $\verlist$-coloring, then $\mxcli{G}\le |\colorset|$.
	\end{observation}
	
	A graph is \emph{split} if its vertex set can be partitioned into a clique and an independent set.
	A graph is a \emph{cograph} (or a $P_4$\emph{-free} graph) if it contains no induced path with four vertices.
	
	\subsection{Modules and modular decomposition}
	\label{sec:MD}
	
	\begin{figure}[t]
		\begin{center}
			\includegraphics[width=0.85\textwidth,clip]{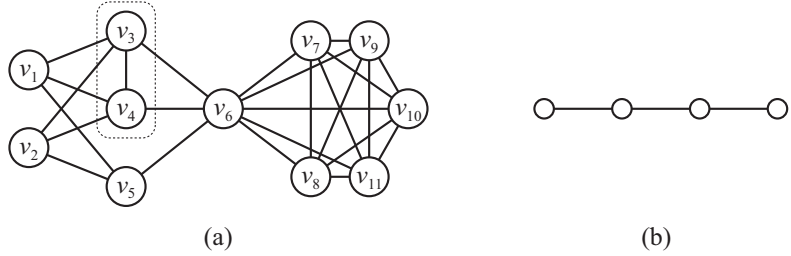}
		\end{center}
		\caption{(a) An example of module and (b) a prime.}
		\label{fig:module}
	\end{figure}
	\begin{figure}[t]
		\begin{center}
			\includegraphics[width=\textwidth,clip]{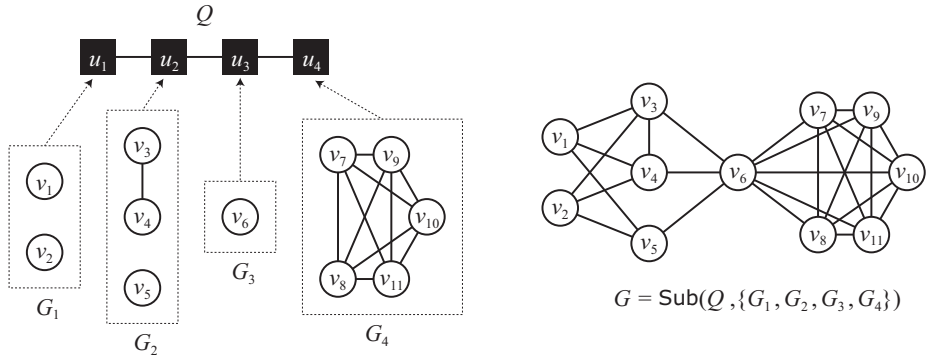}
		\end{center}
		\caption{An example of substitution operation.}
		\label{fig:sub}
	\end{figure}
	
	A \emph{module} of a graph $G=(V,E)$ is a vertex subset $M\subseteq V$ such that $\Neigh{G}{v}\setminus M=\Neigh{G}{w}\setminus M$ for every two vertices $v$ and $w$ in $M$.
	In other words, the module $M$ is the set of vertices whose neighborhoods in $G\setminus M$ are the same.
	For example, the graph in \figurename~\ref{fig:module}(a) has a module $M=\hset{v_3,v_4}$ for which $\Neigh{G}{v_3}\setminus M=\Neigh{G}{v_4}\setminus M=\hset{v_1,v_2,v_6}$ holds.
	Note that the vertex set $V$ of $G$, the set consisting of only a single vertex, and the empty set $\emptyset$ are all modules of $G$; they are called \emph{trivial}.
	A graph $G$ is a \emph{prime} if all of its modules are trivial; for an example, see \figurename~\ref{fig:module}(b).
	
	\begin{figure}[t]
		\begin{center}
			\includegraphics[width=\textwidth,clip]{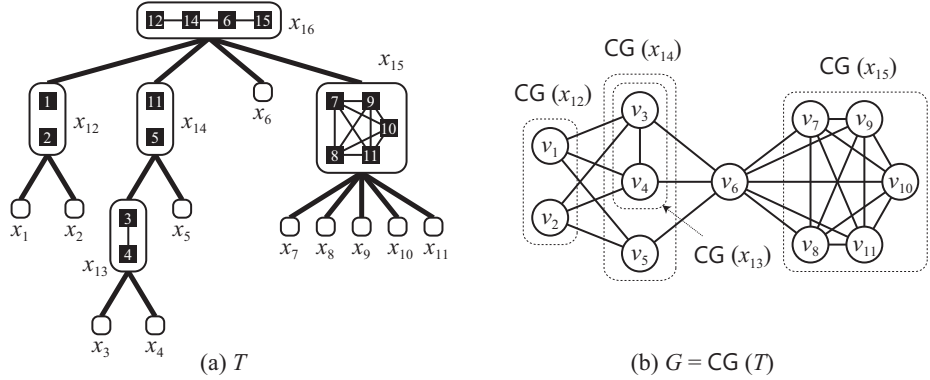}
		\end{center}
		\caption{(a) A substitution tree $T$ for (b) a graph $G$.}
		\label{fig:MDtree}
	\end{figure}
	
	\medskip
	We now introduce the notion of modular decomposition, which was first presented by Gallai in 1967 as a graph decomposition technique~\cite{Gallai67}.
	For a survey, see, e.g.,~\cite{HP10}.
	
	We first define the \emph{substitution} operation, which constructs one graph from more than one graphs.
	Let $Q$ be a graph, called a {\em quotient graph}, consisting of $p$ $(\ge 2)$ nodes $u_1,u_2,\ldots,u_p$, and let $\seq{F}=\hset{G_1,G_2,\ldots,G_p}$ be a family of vertex-disjoint graphs such that $G_i$ corresponds to $u_i$ for every $i\in \hset{1,2,\ldots,p}$.
	The $Q$\emph{-substitution} of $\seq{F}$, denoted by $\SUB{Q}{\seq{F}}$,  is the graph which is obtained by taking a union of all graphs in $\seq{F}$ and then connecting every pair of vertices $v\in V(G_i)$ and $w\in V(G_j)$ by an edge if and only if $u_i$ and $u_j$ are adjacent in $Q$.
	That is, the vertex set of $\SUB{Q}{\seq{F}}$ is $\bigcup \hset{V(G_i)\colon G_i\in \seq{F}}$, and the edge set of $\SUB{Q}{\seq{F}}$ is the union of $\bigcup \hset{E(G_i)\colon G_i\in \seq{F}}$ and $\hset{vw \colon v\in V(G_i),w \in V(G_j), u_i u_j \in E(Q)}$.
	(See \figurename~\ref{fig:sub} as an example.)
	
	A \textit{substitution tree} is a rooted tree $T$ such that each non-leaf node $x\in V(T)$ is associated with a quotient graph $\Quo{x}$ and has $|V(\Quo{x})|$ child nodes.
	For each node $x\in V(T)$, we can recursively define the \emph{corresponding graph} $\coG{x}$ as follows:
	If $x$ is a leaf, $\coG{x}$ consists of a single vertex.
	Otherwise,  let $y_1,y_2,\ldots,y_p$ be $p=|V(\Quo{x})|$ children of $x$, then $\coG{x}=\SUB{\Quo{x}}{\hset{\coG{y_1},\coG{y_2},\ldots,\coG{y_p}}}$.
	For the root $r$ of $T$, $\coG{r}$ is called the \emph{corresponding graph of} $T$, and we denote $\coG{T}:=\coG{r}$.
	We say that $T$ is a substitution tree \emph{for} a graph $G$ if $\coG{T}=G$, and refer to a \emph{node} in $T$ in order to distinguish it from a vertex in $G$.
	Figure~\ref{fig:MDtree}(a) illustrates a substitution tree for the graph $G$ in \figurename~\ref{fig:MDtree}(b);
	each leaf $x_i$, $i\in \hset{1,2,\ldots,11}$, corresponds to the subgraph of $G$ consisting of a single vertex $v_i$.
	We note that the vertex set $V(\coG{x})$ of each corresponding graph $\coG{x}$, $x\in V(T)$, forms a module of $\coG{T}$.
	
	A \emph{modular decomposition tree} $T$ (an \emph{MD-tree} for short) for a graph $G$ is a substitution tree for $G$ which satisfies the following three conditions:
	\begin{itemize}
		\item Each node $x\in V(T)$ applies to one of the following three types:
		\begin{itemize}
			\item a \emph{series} node, whose quotient graph $\Quo{x}$ is a complete graph;
			\item a \emph{parallel} node, whose quotient graph $\Quo{x}$ is an edge-less graph; and
			\item a \emph{prime} node, whose quotient graph $\Quo{x}$ is a prime with at least four vertices.
		\end{itemize}
		\item No edge connects two series nodes.
		\item No edge connects two parallel nodes.
	\end{itemize}
	It is known that any graph $G$ has a unique MD-tree with $O(|V(G)|)$ nodes, and it can be computed in time $O(|V(G)|+|E(G)|)$~\cite{MM05}.
	We denote by $\MD{G}$ the unique MD-tree for a graph $G$.
	The \emph{modular-width} $\mw{G}$ of a graph $G$ is the maximum number of children of a prime node in its MD-tree $\MD{G}$.
	The substitution tree $T$ in \figurename~\ref{fig:MDtree}(a) is indeed the MD-tree for the graph $G$ in \figurename~\ref{fig:MDtree}(b), and hence $\mw{G}=4$;
	note that only $x_{16}$ is a prime node in $T$.
	
	\medskip
	We now define a variant of MD-trees, which will make our proofs and analyses simpler.
	A \emph{pseudo modular decomposition tree} $T$ (a \emph{PMD-tree} for short) for a graph $G$ is a substitution tree for $G$ which satisfies the following two conditions:
	\begin{itemize}
		\item Each node $x\in V(T)$ applies to one of the following three types:
		\begin{itemize}
			\item a $2$\emph{-join} node, whose quotient graph $\Quo{x}$ is a complete graph with exactly two vertices;
			\item a \emph{parallel} node, whose quotient graph $\Quo{x}$ is an edge-less graph; and
			\item a \emph{prime} node, whose quotient graph $\Quo{x}$ is a prime with at least four vertices.
		\end{itemize}
		\item No edge connects two parallel nodes.
	\end{itemize}
	\begin{proposition}
		\label{prop:RMD}
		For any graph $G$, there exists a PMD-tree $T$ with $O(|V(G)|)$ nodes such that each prime node $x\in V(T)$ has at most $\mw{G}$ children, and it can be constructed in polynomial time.
	\end{proposition}
	\begin{proof}
		Recall that an MD-tree $\MD{G}$ for a graph $G$ can be constructed in linear time. 
		Given an MD-tree $\MD{G}$ for a graph $G$, we thus construct a PMD-tree $T$ such that $\coG{T}=\coG{\MD{G}}$ as follows.
		For each series node $x$ of $\MD{G}$ having $m$ $(\ge 3)$ children $y_1,y_2,\ldots,y_m$, we replace it with a binary tree consisting of $m-1$ nodes $x_1,x_2,\ldots,x_{m-1}$ such that $x_i$ has two children $y_i$ and $x_{i+1}$ for each $i\in \hset{1,2,\ldots,m-2}$ and $x_{m-1}$ has two children $y_{m-1}$ and $y_m$.
		A quotient graph $\Quo{x_i}$ of each new node $x_i$ is defined as a complete graph with exactly two vertices.
		Then, $T$ is a PMD-tree for $G$, it has at most $O(|V(G)|)$ nodes, and each prime node $x\in V(T)$ has at most $\mw{G}$ children.
		Moreover, this process can be done in time polynomial in $|V(\MD{G})|=O(|V(G)|)$.
		\qed
	\end{proof}
	
	\medskip
	We denote by $\PMD{G}$ a substitution tree for $G$ such that each prime node $x\in V(T)$ has at most $\mw{G}$ children.
	The \emph{pseudo modular-width} $\pmw{G}$ of a graph $G$ is the maximum number of children of a non-parallel node in its PMD-tree.
	Notice that $\pmw{G}=\max \hset{2,\mw{G}}$ holds.

	\subsection{Other notation}
	
	Let $G$ be a graph, and let $\verlist \colon V(G)\to 2^\colorset$ be a list.
	For two $\verlist$-colorings $\mapf$ and $\mapf^\prime$ of a graph $G=(V,E)$, we define the \emph{difference} $\diff{\mapf}{\mapf^\prime}$ \emph{between} $\mapf$ \emph{and} $\mapf^\prime$ as the set $\hset{v \in V \colon \mapf(v) \neq \mapf^\prime(v)}$.
	Notice that $\mapf$ and $\mapf^\prime$ are adjacent if and only if $|\diff{\mapf}{\mapf^\prime}|=1$.
	
	We express an instance $\Inst$ of \lcr ~by a $4$-tuple $(G,\verlist,\mapf_{\ini},\mapf_{\tar})$ consisting of a graph $G$, a list $\verlist$, and \emph{initial} and \emph{target} $\verlist$-colorings $\mapf_{\ini}$ and $\mapf_{\tar}$ of $G$.
	
	Finally, we introduce a notion of ``restriction'' of mappings and instances.
	Consider an arbitrary mapping $\mu \colon V(G)\to S$, where $G$ is a graph and $S$ is any set.
	For a subgraph $H$ of $G$, we denote by $\rest{\mu}{H}$ the \emph{restriction} of $\mu$ on $V(H)$, that is, $\rest{\mu}{H}$ is a mapping from $V(H)$ to $S$ such that $\rest{\mu}{H}(v)=\mu(v)$ for each vertex $v\in V(H)$.
	Let $\Inst=(G,\verlist,\mapf_{\ini},\mapf_{\tar})$ be an instance of \lcr.
	For a subgraph $H$ of $G$, we define the \emph{restriction} $\rest{\Inst}{H}$ of $\Inst$ (on $H$) as the instance $(H,\rest{\verlist}{H},\rest{\mapf_{\ini}}{H},\rest{\mapf_{\tar}}{H})$ of \lcr.
	Notice that $\rest{\mapf_{\ini}}{H}$ and $\rest{\mapf_{\tar}}{H}$ are proper $\rest{\verlist}{H}$-colorings of $H$.

	%
	%
	
	
	\section{Fixed-Parameter Algorithm for Bounded Modular-Width Graphs}
	\label{sec:FPT}
	
	The following is our main theorem of this section.
	\begin{theorem}
	\label{the:bmw}
	\textsc{List coloring reconfiguration} is fixed-parameter tractable when parameterized by $k+\MW$, where $k$ and $\MW$ are the upper bounds on the size of the color set and the modular-width of an input graph, respectively.
	\end{theorem}

	Because it is known that any cograph has modular-width zero, we have the following result as a corollary of Theorem~\ref{the:bmw}. 
	\begin{corollary}
	\label{cor:cograph}
	\textsc{List coloring reconfiguration} is fixed-parameter tractable for cographs when parameterized by the size $k$ of the color set. 
	\end{corollary}	
	
	Recall that $\pmw{G}=\max \hset{2,\mw{G}}$, and hence $\pmw{G} \le \mw{G}+2$.
	Therefore, as a proof of Theorem~\ref{the:bmw}, it suffices to give a fixed-parameter algorithm for \lcr ~with respect to $k+\PMW$, where $\PMW$ is an upper bound on $\pmw{G}$. 
	\subsection{Reduction rule}
\label{sec:reduce}

\begin{figure}[t]
	\begin{center}
		\includegraphics[width=0.75\textwidth,clip]{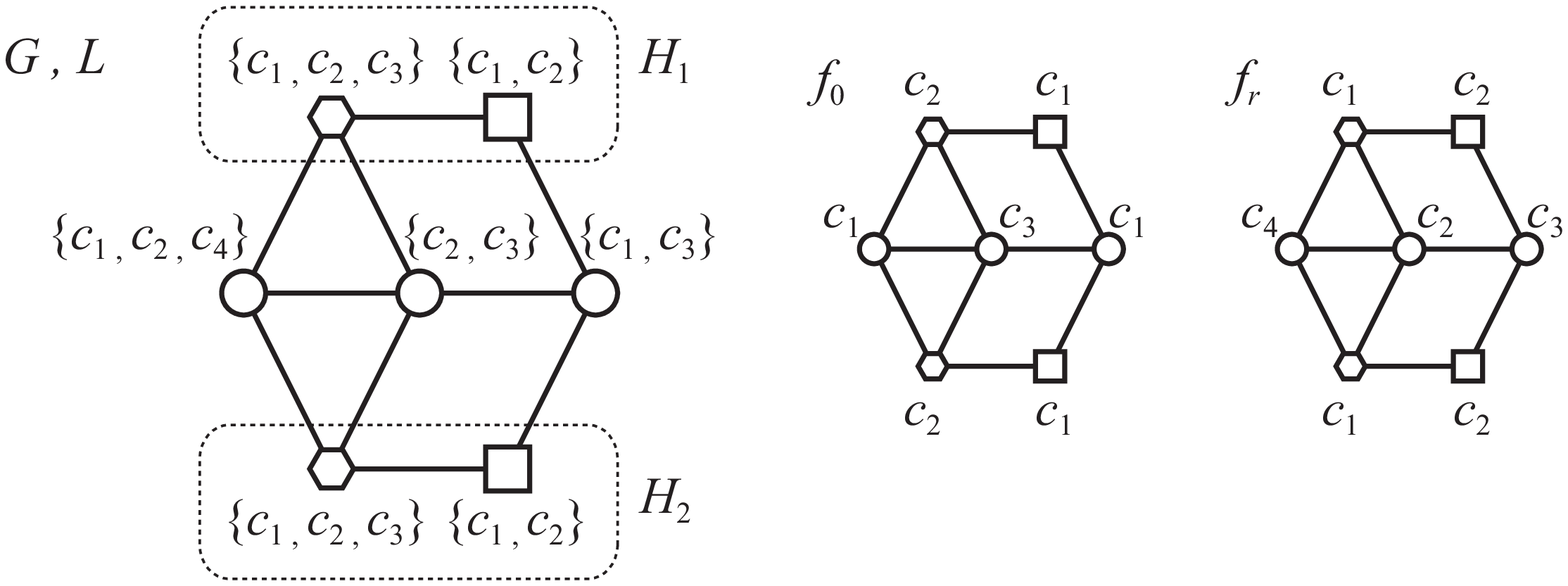}
	\end{center}
	\caption{An instance $\Inst=(G,\verlist,\mapf_{\ini},\mapf_{\tar})$ of \lcr , and two identical subgraphs $H_1$ and $H_2$.}
	\label{fig:identical}
\end{figure}

	In this subsection, we give a useful lemma, which compresses an input graph into a smaller graph with keeping the reconfigurability.

	Let $\Inst=(G,\verlist,\mapf_{\ini},\mapf_{\tar})$ be an instance of \lcr.
	For each vertex $v\in V(G)$, we define a \emph{vertex assignment} $\asgn{v}$ as a triple $(\verlist(v),\mapf_{\ini}(v),\mapf_{\tar}(v))$ consisting of a list, and initial and target color assignments of $v$.
	Let $H_1$ and $H_2$ be two induced subgraphs of $G$ such that $|V(H_1)|=|V(H_2)|$ and $V(H_1)\cap V(H_2)=\emptyset$
	Then, $H_1$ and $H_2$ are \emph{identical} (on $\Inst$) if there exists a bijective function $\phi \colon V(H_1) \to V(H_2)$ which satisfies all the following two conditions:
	\begin{enumerate}
		\item $H_1$ and $H_2$ are isomorphic under $\phi$, that is, $vw\in E(H_1)$ if and only if $\phi(v)\phi(w)\in E(H_2)$.
		\item For all vertices $v\in V(H_1)$,
		\begin{enumerate}
			\item $N(G,v)\setminus V(H_1)=N(G,\phi(v))\setminus V(H_2)$; and
			\item $\asgn{v}=\asgn{\phi(v)}$, that is, $\verlist(v)=\verlist(\phi(v))$, $\mapf_{\ini}(v)=\mapf_{\ini}(\phi(v))$ and $\mapf_{\tar}(v)=\mapf_{\tar}(\phi(v))$.
		\end{enumerate}
	\end{enumerate}

	\noindent
	We note that the condition 2(a) implies that there is no edge between $H_1$ and $H_2$.
	Figure~\ref{fig:identical} shows an example of identical subgraphs $H_1$ and $H_2$ on $\Inst=(G,\verlist,\mapf_{\ini},\mapf_{\tar})$, where the bijective function maps each vertex in $H_1$ to a vertex in $H_2$ with the same shape.
	
	We now prove the following key lemma, which holds for any graph.
	\begin{lemma}
	\label{lem:reduce}
	{\rm (Reduction rule)}
	Let $\Inst=(G,\verlist,\mapf_{\ini},\mapf_{\tar})$ be an instance of \lcr, and let $H_1$ and $H_2$ be two identical subgraphs of $G$.
	Then, $\rest{\Inst}{G\setminus H_2}$ is a yes-instance if and only if $\Inst$ is.
	\end{lemma}
	\begin{proof}
		We assume that $H_1$ and $H_2$ are identical under a bijective function $\phi\colon V(H_1)\to V(H_2)$, and let $G^\prime=G\setminus H_2$.
		
		We first prove the if direction.
		Suppose that $\Inst$ is a yes-instance.
		Then, there exists a reconfiguration sequence $\hseq{\mapf_0, \mapf_1, \ldots, \mapf_\ell}$ for $\Inst$, where $\mapf_\ell=\mapf_\tar$.
		For each $i\in \hset{0,1,\ldots,\ell}$, since $\mapf_i$ is an $\verlist$-coloring of $G$, $\rest{\mapf_i}{G^\prime}$ is an $\rest{\verlist}{G^\prime}$-coloring of $G^\prime$.
		Therefore, by removing all consecutive duplicate $\rest{\verlist}{G^\prime}$-colorings, $\hseq{\rest{\mapf_0}{G^\prime}, \rest{\mapf_1}{G^\prime}, \ldots,\rest{\mapf_{\ell^\prime}}{G^\prime}}$ is a reconfiguration sequence for $\rest{\Inst}{G^\prime}$.
		Thus $\rest{\Inst}{G^\prime}$ is a yes-instance.
		
		We now prove the only-if direction.
		Suppose that $\rest{\Inst}{G^\prime}$ is a yes-instance.
		Then, there exists a reconfiguration sequence $\seq{S}^\prime=\hseq{\mapg_0, \mapg_1, \ldots, \mapg_\ell}$ for $\rest{\Inst}{G^\prime}$ with $\mapg_0=\rest{\mapf_\ini}{G^\prime}$ and $\mapg_\ell=\rest{\mapf_\tar}{G^\prime}$.
		Our goal is to construct a reconfiguration sequence $\seq{S}$ for $\Inst$ from $\seq{S}^\prime$.
		For each $i\in \hset{0,1,\ldots,\ell}$, we first extend $\mapg_i$ to $\widehat{\mapf}_i$ as follows:
		\[
		\widehat{\mapf}_i(v) =\left\{
		\begin{array}{ll}
		\mapg_i(\phi^{-1}(v)) & ~~~\mbox{if $v\in V(H_2)$}; \\
		\mapg_i(v) & ~~~\mbox{otherwise}.
		\end{array} \right.
		\]
		We claim that $\widehat{\mapf}_i$ is a proper $\verlist$-coloring of $G$.
		To show this, it suffices to check that $\widehat{\mapf}_i(v)\ne \widehat{\mapf}_i(w)$ holds for each $v\in V(H_2)$ and its neighbors $w\in \Neigh{G}{v}$.
		If $w\in V(H_2)$, $\widehat{\mapf}_i(v)=\mapg_i(\phi^{-1}(v))\ne \mapg_i(\phi^{-1}(w))=\widehat{\mapf}_i(w)$ holds because $\phi^{-1}(v)\phi^{-1}(w)\in E(G^\prime)$ and $\mapg_i$ is an $\rest{\verlist}{G^\prime}$-coloring.
		Otherwise, $\widehat{\mapf}_i(v)=\mapg_i(\phi^{-1}(v))\ne \mapg_i(w)=\widehat{\mapf}_i(w)$ holds because $\phi^{-1}(v)w\in E(G^\prime)$ and $\mapg_i$ is an $\rest{\verlist}{G^\prime}$-coloring.
		Therefore, the obtained sequence $\widehat{\Seq}$ consists only of $\verlist$-colorings of $G$.
		However, there may exist several indices $i\in \hset{0,1,\ldots,\ell-1}$ such that $\widehat{\mapf}_i$ and $\widehat{\mapf}_{i+1}$ are not adjacent, because $|\diff{\widehat{\mapf}_i}{\widehat{\mapf}_{i+1}}|>1$ may hold.
		Recall that $\mapg_i$ and $\mapg_{i+1}$ are adjacent for each $i\in \hset{0,1,\ldots,\ell-1}$, that is, $\diff{\mapg_i}{\mapg_{i+1}}=\hset{w}$ for some vertex $w\in V(G^\prime)$.
		If $w\notin V(H_1)$, $\diff{\widehat{\mapf}_i}{\widehat{\mapf}_{i+1}}=\hset{w}$ holds, and hence $\widehat{\mapf}_i$ and $\widehat{\mapf}_{i+1}$ are adjacent.
		Otherwise, $\diff{\widehat{\mapf}_i}{\widehat{\mapf}_{i+1}}=\hset{w,\phi(w)}$ holds, and hence $\widehat{\mapf}_i$ and $\widehat{\mapf}_{i+1}$ are not adjacent.
		In this case, between $\widehat{\mapf}_i$ and $\widehat{\mapf}_{i+1}$, we insert an $\verlist$-coloring $\widetilde{\mapf}_i$ of $G$ defined as follows:
		\[
		\widetilde{\mapf}_i(v) =\left\{
		\begin{array}{ll}
		\widehat{\mapf}_{i+1}(v) & ~~~\mbox{if $v=w$}; \\
		\widehat{\mapf}_i(v) & ~~~\mbox{if $v=\phi(w)$}; \\
		\widehat{\mapf}_i(v) & ~~~\mbox{otherwise}.
		\end{array} \right.
		\]
		Observe that $\widetilde{\mapf}_i$ is a proper $\verlist$-coloring of $G$.
		Moreover, both $\diff{\widehat{\mapf}_i}{\widetilde{\mapf}_i}=\hset{w}$ and $\diff{\widetilde{\mapf}_i}{\widehat{\mapf}_{i+1}}=\hset{\phi(w)}$ hold.
		Thus, we obtain a proper reconfiguration sequence $\seq{S}$ for $\Inst$ as claimed.
		\qed
	\end{proof}
	\subsection{Kernelization}
\label{sec:bmw}
	
	Let $\Inst=(G,\verlist,\mapf_{\ini},\mapf_{\tar})$ be an instance of \lcr.
	Suppose that the color set $C$ has at most $k$ colors, $G$ is a connected graph with $\pmw{G}\le \PMW$, and all vertices of $G$ are totally ordered according to an arbitrary binary relation $\prec$.
	
\subsubsection{Sufficient condition for identical subgraphs.}
	
	We first give a sufficient condition for which two nodes in a PMD-tree $\PMD{G}$ for $G$ correspond to identical subgraphs.
	Let $x\in V(\PMD{G})$ be a node, let $p:=|V(\coG{x})|$, and assume that all vertices in $V(\coG{x})$ are labeled as $v_1,v_2,\ldots,v_p$ according to $\prec$; that is, $v_i\prec v_j$ holds for each $i,j$ with $1\le i<j\le p$.
	Let $m \ge p$ be some integer which will be defined later.
	We now define an $(m+1)\times m$ matrix $\IDM{x}$ as follows:
	\[
	(\IDM{x})_{i,j}=\left\{
			\begin{array}{ll}
			1 & ~~~\mbox{if $i,j\le p$ and $v_i v_j\in E(\coG{x})$}; \\
			0 & ~~~\mbox{if $i,j\le p$ and $v_i v_j\notin E(\coG{x})$}; \\
			0 & ~~~\mbox{if $p < i \le m$ or $p < j \le m$}; \\
			\asgn{v_j} & ~~~\mbox{if $i=m+1$ and $j\le p$}; \\
			\emptyset & ~~~\mbox{otherwise},
			\end{array} \right.
	\]
	where $(\IDM{x})_{i,j}$ denotes an $(i,j)$-element of $\IDM{x}$.
	Notice that $\IDM{x}$ contains an adjacency matrix of $\coG{x}$ at its upper left $p\times p$ submatrix, and the bottommost row represents the vertex assignment of each vertex in $V(\coG{x})$.
	We call $\IDM{x}$ an $m$\emph{-ID-matrix} of $x$.
	For example, consider the node $x_{13}$ in \figurename~\ref{fig:MDtree}(a).
	Then, $p=2$, and a $4$-ID-matrix of $x_{13}$ is as follows:
	\begin{equation*}
		\IDMi{x_{13}}{4}=
		\begin{bmatrix}
			0 & 1 & 0 & 0 \\
			1 & 0 & 0 & 0 \\
			0 & 0 & 0 & 0 \\
			0 & 0 & 0 & 0 \\
			\asgn{x_3} & \asgn{x_4} & ~~\emptyset~~~ & ~~\emptyset~~~	
		\end{bmatrix}
	\end{equation*}
	\begin{lemma}
	\label{lem:IDM}
	Let $y_1$ and $y_2$ be two children of a parallel node $x$ in $\PMD{G}$, and let $m$ be an integer with $m\ge \max \hset{|V(\coG{y_1})|,|V(\coG{y_2})|}$.
	If $\IDM{y_1}=\IDM{y_2}$ holds, then $\coG{y_1}$ and $\coG{y_2}$ are identical.
	\end{lemma}
	\begin{proof}
		Let $p_1:=|V(\coG{y_1})|$ and $p_2:=|V(\coG{y_2})|$.
		Observe that $(\IDM{y_1})_{m+1,j}\allowbreak \ne \emptyset$ if and only if $j\le p_1$, and $(\IDM{y_2})_{m+1,j}\ne \emptyset$ if and only if $j\le p_2$.
		By the assumption that $(\IDM{y_1})_{m+1,j}=(\IDM{y_2})_{m+1,j}$ for all $j\in \hset{1,2,\ldots,m}$, we have $p_1=p_2$; we denote by $p$ this value.
		
		We now check that $\coG{y_1}$ and $\coG{y_2}$ are identical.
		The condition~1 of identical subgraphs holds, because the upper left $p \times p$ submatrices in $\IDM{y_1}$ and $\IDM{y_2}$ correspond to the adjacency matrices of $\coG{y_1}$ and $\coG{y_2}$, respectively.
		The condition~2(b) holds, because the bottommost rows are the same in $\IDM{y_1}$ and $\IDM{y_2}$.
		Finally, we claim that the condition~2(a) holds, as follows.
		Since $x$ is a parallel node, $N(G,v)\setminus V(\coG{y_1})=N(G,v)\setminus V(\coG{x})$ holds for all vertices $v$ in $\coG{y_1}$.
		Similarly, $N(G,w)\setminus V(\coG{y_2})=N(G,w)\setminus V(\coG{x})$ holds for all vertices $w$ in $\coG{y_2}$.
		Recall that $V(\coG{x})$ is a module of $G$, that is, $N(G,v)\setminus V(\coG{x})=N(G,v^\prime)\setminus V(\coG{x})$ holds for any vertice $v,v^\prime \in V(\coG{x})$.
		Therefore, $N(G,v)\setminus V(\coG{y_1})=N(G,w)\setminus V(\coG{y_2})$ holds any pair of $v \in V(\coG{y_1})$ and $w \in  V(\coG{y_2})$.
		Thus, the condition~2(a) holds. 
		\qed
	\end{proof}
	
\subsubsection{Kernelization algorithm.}
	
	We now describe how to kernelze an input instance.
	(See \figurename~\ref{fig:alg} as an example.)
	Our algorithm traverses a PMD-tree $\PMD{G}$ of $G$ by a depth-first search in post-order starting from the root of $\PMD{G}$, that is, the algorithm processes a node of $\PMD{G}$ after its all children are processed.
	
	\begin{figure}[h!]
		\begin{center}
			\includegraphics[width=0.8\textwidth,clip]{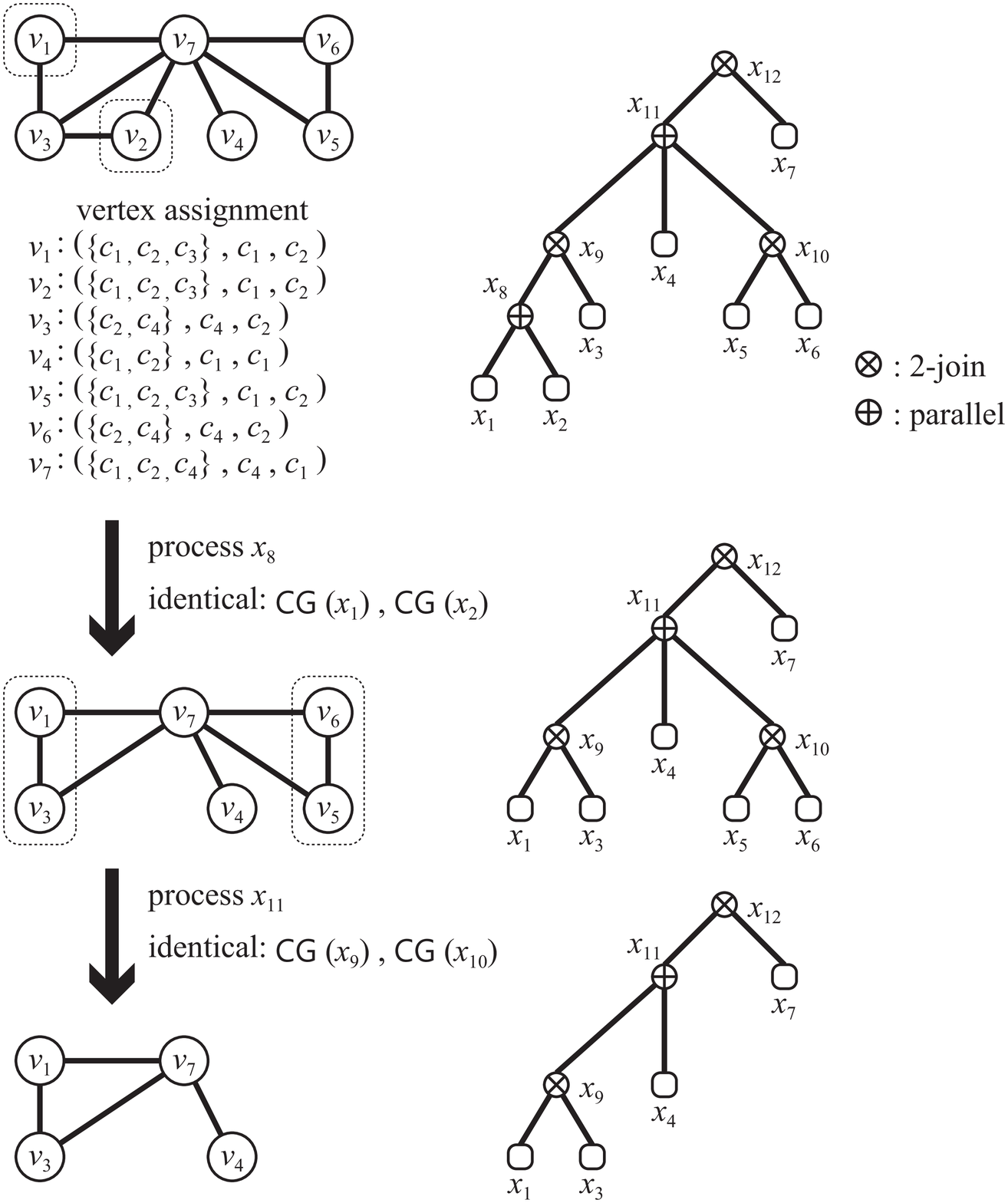}
		\end{center}
		\caption{An example of an application of our algorithm.
			We first focus on $x_8$, which is a parallel node whose children are already kernelized, and find that $\IDMi{x_1}{1}=\IDMi{x_2}{1}$ holds.
			Therefore, we delete $\coG{x_1}$ from the input graph.
			Then, $x_8$ has only one child now, and hence we contract an edge $x_8 x_9$ in order to maintain being a PMD-tree.
			We next focus on $x_{11}$ and find that $\IDMi{x_9}{2}=\IDMi{x_{10}}{2}$ holds.
			After removing $\coG{x_{10}}$ from the current graph and fixing a PMD-tree, we have processed all parallel nodes.}
		\label{fig:alg}
	\end{figure}
	
	Let $x\in V(\PMD{G})$ be a node which is currently visited.
	If $x$ is a non-parallel node, we do nothing.
	Otherwise (i.e., if $x$ is a parallel node,) let $Y$ be the set of all children of $x$, and let $m:=\max_{y\in Y}|V(\coG{y})|$.
	We first construct $m$-ID-matrices of all children of $x$.
	If there exist two nodes $y_1$ and $y_2$ such that $\IDM{y_1}=\IDM{y_2}$, then $\coG{y_1}$ and $\coG{y_2}$ are identical; and hence we remove $\coG{y_2}$ from $G$ by Lemma~\ref{lem:reduce}.
	Then, we modify $\PMD{G}$ in order to keep it still being a PMD-tree for the resulting graph as follows.
	We remove a subtree rooted at $y_2$ from $\PMD{G}$, and delete a node corresponding to $y_2$ from a quotient graph $\Quo{x}$ of $x$.	
	If this removal makes $x$ having only one child $y$ in the PMD-tree, we contract the edge $xy$ into a new node $x^\prime$ such that $\Quo{x^\prime}=\Quo{x}$.
		
	The running time of this kernelization can be estimated as follows.
	For each node $x\in V(\PMD{G})$, the construction of $m$-ID-matrices can be done in time $O(|Y|\cdot m^2)=O(|V(G)|^3)$.
	We can check if $\IDM{y_1}=\IDM{y_2}$ for each pair of children $y_1$ and $y_2$ of $x$ in time $O(m^2)=O(|V(G)|^2)$.
	Moreover, a modification of $\PMD{G}$, which follows an application of Lemma~\ref{lem:reduce}, can be done in polynomial time.
	Recall that the number of children of $x$ and the size of a PMD-tree $\PMD{G}$ are both bounded linearly in $|V(G)|$, and hence our kernelization can be done in polynomial time.


\subsubsection{Size of the kernelized instance.}
	
	We finally prove that the size of the obtained instance $\Inst^\prime=(G^\prime,\verlist^\prime,\mapf^\prime_{\ini},\mapf^\prime_{\tar})$ depends only on $k+\PMW$;
	recall that $\PMW$ is the upper bound on $\pmw{G}$.
	By Observation~\ref{obs:colors}, we can assume that the maximum clique size $\mxcli{G^\prime}$ is at most $k$.
	In addition, $G^\prime$ is connected since $G$ is connected and an application of Lemma~\ref{lem:reduce} does not affect the connectivity of the graph.
	Therefore, it suffices to prove the following lemma.
	\begin{lemma}
	\label{lem:kernel}
	The graph $G^\prime$ has at most $\kernel{\mxcli{G^\prime}}$ vertices, where $\kernel{i}$ is recursively defined for an integer $i\ge 1$ as follows{\rm :}
		\[
			\kernel{i} =\left\{
					\begin{array}{ll}
					1 & ~~~\mbox{if $i=1$}; \\
					\PMW \cdot \kernel{i-1} \cdot \sqrt{2}^{(\kernel{i-1})^2}\cdot (2^k \cdot k^2)^{\kernel{i-1}} & ~~~\mbox{otherwise}.
					\end{array} \right.
		\]
	In particuler, $\kernel{\mxcli{G^\prime}}$ depends only on $k+\PMW$.
	\end{lemma}
	\begin{proof}
		We prove the lemma by induction on $\mxcli{G^\prime}$.
		If $\mxcli{G^\prime}=1$, then we have $|V(G^\prime)|=1=\kernel{1}$ since $G^\prime$ is connected.
		
		We thus assume in the remainder of the proof that $\mxcli{G^\prime}>1$.
		Then, the root $r$ of a PMD-tree for $G^\prime$ must be a non-parallel node since $G^\prime$ is connected.
		Because $r$ has at most $\pmw{G^\prime}\le \PMW$ children, it suffices to show that the corresponding graph of each child of $r$ has at most
		\begin{eqnarray*}
			\kernel{\mxcli{G^\prime}-1} \cdot \sqrt{2}^{(\kernel{\mxcli{G^\prime}-1})^2}\cdot (2^k \cdot k^2)^{\kernel{\mxcli{G^\prime}-1}}
		\end{eqnarray*}
		vertices.
		We will prove this by showing the following two claims for any child $x$ of $r$:
		\begin{enumerate}[(A)]
			\item $|V(H)|\le \kernel{\mxcli{G^\prime}-1}$ holds for any connected component $H$ of $\coG{x}$; and
			\item $\coG{x}$ has at most $\sqrt{2}^{(\kernel{\mxcli{G^\prime}-1})^2}\cdot (2^k \cdot k^2)^{\kernel{\mxcli{G^\prime}-1}}$ connected components.
		\end{enumerate}
		
		In order to prove the claim (A), we first claim that $\mxcli{\coG{x}}<\mxcli{G^\prime}$ holds.
		Assume for a contradiction that $\coG{x}$ contains a clique $X$ of size $\mxcli{G^\prime}$.
		Let $\hat{x}$ be a node of a quotient graph $\Quo{r}$ which corresponds to $x$.
		By the definition, $\Quo{r}$ is connected, and hence there exists a node $\hat{y}\in V(\Quo{r})$ which is adjacent to $\hat{x}$.
		Let $y\in \PMD{G}$ be the child of $r$ corresponding to $\hat{y}$.
		Recall that all vertices in $X$ are connected with at least one vertex $v$ in $V(\coG{y})$ by the substitution operation, which means that $G^\prime$ has a clique $X\cup \hset{v}$ of size $\mxcli{G^\prime}+1$.
		This contradicts the assumption that the maximum clique size of $G^\prime$ is $\mxcli{G^\prime}$; this completes the proof of the claim.
		Note that $\mxcli{H}\le \mxcli{\coG{x}}<\mxcli{G^\prime}$ holds for any connected component $H$ of $\coG{x}$.
		Therefore, $|V(H)|\le \kernel{\mxcli{G^\prime}-1}$ follows from the induction hypothesis.
		
		We next prove the claim (B).
		If $x$ is a non-parallel node, $\coG{x}$ is connected and hence we are done.
		The remaining case is where $x$ is a parallel node.
		Let $H$ be a connected component of $\coG{x}$, and let $Y$ be the set of all children of $x$.
		Then, there exists exactly one child $y\in Y$ such that $V(\coG{y})\supseteq V(H)$.
		Since a PMD-tree has no edge joining two parallel nodes, $y$ is not a parallel node.
		Thus, $\coG{y}$ is connected, and hence we indeed have $V(\coG{y})=V(H)$.
		Therefore, it suffices to bound the size of $Y$ instead of the number of connected components in $\coG{x}$.
		Let $m:=\max_{y\in Y}|V(\coG{y})|$.	
		Since $G^\prime$ is already kernelized, $\IDM{y_1}\ne \IDM{y_2}$ holds for any two children $y_1,y_2\in Y$.
		Therefore, $|Y|$ cannot exceed the number of distinct $m$-ID-matrices.
		Recall that the upper $m\times m$ submatrix consists of $m^2$ values from $\hset{0,1}$, its $(i,i)$-element is $0$ for each $i\in \hset{1,2,\ldots,m}$, and it is symmetric.
		Therefore, the number of such $m\times m$ submatrices can be bounded by $2^{m^2/2}=\sqrt{2}^{m^2}$.
		Recall that all elements of the $(m+1)$-st row are chosen from the set $2^C\times C\times C$, where $C$ is the color set of size at most $k$.
		Therefore, the number of such $1 \times m$ submatrices can be bounded by $(2^k \cdot k^2)^m$.
		By the claim (A), we have $m=\max_{y\in Y}|V(\coG{y})|\le \kernel{\mxcli{G^\prime}-1}$.
		Therefore, the size of $Y$, and hence the number of connected components in $\coG{x}$, can be bounded by 
		\begin{eqnarray*}
			\sqrt{2}^{(\kernel{\mxcli{G^\prime}-1})^2}\cdot (2^k \cdot k^2)^{\kernel{\mxcli{G^\prime}-1}}.
		\end{eqnarray*}
		
		From the claims (A) and (B), we have the following inequality.
		\begin{eqnarray*}
			|V(G^\prime)|\le w \times \kernel{\mxcli{G^\prime}-1} \times \sqrt{2}^{(\kernel{\mxcli{G^\prime}-1})^2}\cdot (2^k \cdot k^2)^{\kernel{\mxcli{G^\prime}-1}}
		\end{eqnarray*}
		as claimed.
		In particular, we can conclude that $\kernel{\mxcli{G^\prime}}$ depends only on $k+\PMW$, because $\mxcli{G^\prime}\le k$.
		\qed
	\end{proof}
	
	Finally, we prove Theorem~\ref{the:bmw}.
	By the above discussions, we can compute the kernelized instance $\Inst^\prime=\rest{\Inst}{G^\prime}$ of \lcr~in polynomial time.
	Because the size of $\Inst^\prime$ depends only on $k+\PMW$, we can solve $\Inst^\prime$ by enumerating all $\rest{\verlist}{G^\prime}$-colorings.
	The running time for this enumeration depends only on $k+\PMW$, and hence we obtain a fixed-parameter algorithm for \lcr.
	
	This completes the proof of Theorem~\ref{the:bmw}.

%
%
	\section{Shortest Variant}
\label{sec:shortest}

	In this section, we study the shortest variant, \lcsr.
	We note that the shortest length can be expressed by a polynomial number of bits, because there are at most $k^n$ colorings for a graph with $n$ vertices and $k$ colors.
	Therefore, the answer can be output in polynomial time.
	The following is our result.
	\begin{theorem}
	\label{the:split}
	\textsc{List coloring shortest reconfiguration} is fixed-parameter tractable when parameterized by $k+\VC$, where $k$ and $\VC$ are the upper bounds on the sizes of the color set and a minimum vertex cover of an input graph, respectively.
	\end{theorem}
	
	As a corollary, we have the following result.
	\begin{corollary}
		\label{cor:split}
	\textsc{List coloring shortest reconfiguration} is fixed-parameter tractable for split graphs when parameterized by the size $k$ of the color set. 
	\end{corollary}
	\begin{proof}
		Let $\Inst=(G,\verlist,\mapf_{\ini},\mapf_{\tar})$ be an instance of \lcsr ~such that $G$ is a split graph.
		Assume that the vertex set of $G$ can be partitioned into a clique $V^\prime$ and an independent set $V^{\prime \prime}$.
		By Observation~\ref{obs:colors}, we have $|V^\prime|\le \mxcli{G}\le k$.
		Observe that $V^\prime$ forms a vertex cover of $G$.
		Thus, $\VC \le |V^\prime|\le k$ holds for split graphs.
		\qed
	\end{proof}

	As a proof of Theorem~\ref{the:split}, we give such a fixed-parameter algorithm.
	Our basic idea is the same as the fixed-parameter algorithm in Section~\ref{sec:FPT}.
	However, in order to compute the shortest length, we consider a more general ``weighted'' version of \lcsr, which is defined as follows.
	Let $\Inst=(G,\verlist,\mapf_{\ini},\mapf_{\tar})$ be an instance of \lcr, and assume that each vertex $v\in V(G)$ has a \emph{weight} $\weight(v)\in \NN$, where $\NN$ is the set of all positive integers.
	For two adjacent $\verlist$-colorings $\mapf$ and $\mapf^\prime$ of a graph $G$, we define the \emph{gap} $\gap{\weight}{\mapf}{\mapf^\prime}$ \emph{between} $\mapf$ \emph{and} $\mapf^\prime$ as the weight $\weight(v)$ of $v$, where $v$ is a unique vertex in $\diff{\mapf}{\mapf^\prime}$.
	The \emph{length} $\len{\weight}{\Seq}$ \emph{of} a reconfiguration sequence $\Seq=\hseq{\mapf_0, \mapf_1,\allowbreak \ldots, \mapf_\ell}$ is defined as $\len{\weight}{\Seq}=\sum_{i=1}^{\ell} \gap{\weight}{\mapf_{i-1}}{\mapf_i}$.
	We denote by $\OPT{\Inst}{\weight}$ the length of a shortest reconfiguration sequence between $\mapf_\ini$ and $\mapf_\tar$; we define $\OPT{\Inst}{\weight}=+\infty$ if $\Inst$ is a no-instance of \lcr.
	Then, \lcsr ~can be seen as computing $\OPT{\Inst}{\weight}$ for the case where every vertex has weight one.
	Thus, to prove Theorem~\ref{the:split}, it suffices to construct a fixed-parameter algorithm for the weighted version when parameterized by $k+\VC$.
	
	As with Section~\ref{sec:FPT}, we again use the concept of kernelization to prove Theorem~\ref{the:split}.
	More precisely, for a given instance $(\Inst,\weight)$, we first construct an instance $(\Inst^\prime=(G^\prime,\verlist^\prime,\mapf_{\ini}^\prime,\mapf_{\tar}^\prime),\weight^\prime)$ in polynomial time such that the size of $\Inst^\prime$ depends only on $k+\VC$, and $\OPT{\Inst}{\weight}=\OPT{\Inst^\prime}{\weight^\prime}$ holds.
	Then, we can compute $\OPT{\Inst^\prime}{\weight^\prime}$ by computing a (weighted) shortest path between $\mapf_{\ini}^\prime$ and $\mapf_{\tar}^\prime$ in an edge-weighted graph defined as follows:
	the vertex set consists of all $\verlist^\prime$-colorings of $G^\prime$, and each pair of adjacent $\verlist^\prime$-colorings are connected by an edge with a weight corresponding to the gap between them.
	
\subsection{Reduction rule for the weighted version}

\begin{figure}[t]
	\begin{center}
		\includegraphics[width=0.8\textwidth,clip]{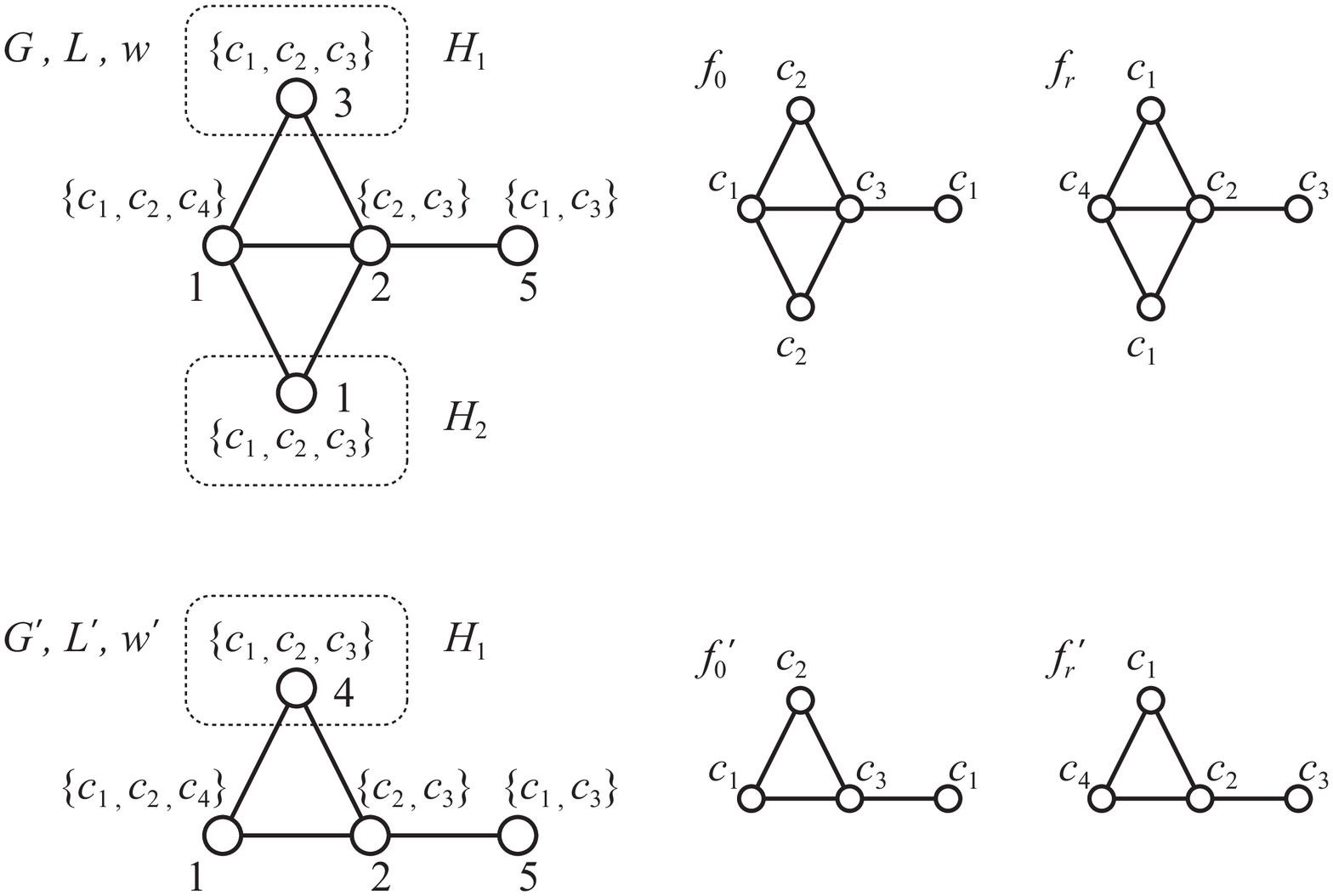}
	\end{center}
	\caption{Two identical subgraphs $H_1$ and $H_2$ for an instance $(\Inst=(G,\verlist,\mapf_{\ini},\mapf_{\tar}),\weight)$, and a new instance $(\Inst^\prime,\weight^\prime)$.}
	\label{fig:widentical}
\end{figure}

	In this subsection, we give the counterpart of Lemma~\ref{lem:reduce} for the weighted version.
	
	We first introduce some notation.
	Let $\Seq=\hseq{\mapf_0, \mapf_1, \ldots, \mapf_\ell}$ be a reconfiguration sequence for an instance $\Inst=(G,\verlist,\mapf_{\ini},\mapf_{\tar})$ of \lcr.
	For each vertex $v\in V(G)$, we denote by $\nrec{\Seq}{v}$ the number of indices $i$ such that $\diff{\mapf_{i-1}}{\mapf_i}=\hset{v}$.
	In other words, $\nrec{\Seq}{v}$ is the number of steps recoloring $v$ in $\Seq$.
	Notice that $\len{\weight}{\Seq}=\sum_{v\in V(G)}\weight(v)\cdot \nrec{\Seq}{v}$ holds for any weight function $w\colon V(G)\to \NN$.
	
	\medskip	
	Let $(\Inst=(G,\verlist,\mapf_{\ini},\mapf_{\tar}),\weight)$ be an instance of the weighted version, and assume that there exist two identical subgraphs $H_1$ and $H_2$ of $G$, both of which consist of single vertices, say, $V(H_1)=\hset{v_1}$ and $V(H_2)=\hset{v_2}$.
	We now define a new instance $(\Inst^\prime,\weight^\prime)$ as follows
	(see also \figurename~\ref{fig:widentical}):
	\begin{itemize}
		\item $\Inst^\prime=\rest{\Inst}{G\setminus H_2}$; and
		\item $\weight^\prime(v_1)=\weight(v_1)+\weight(v_2)$ and $\weight^\prime(v)=\weight(v)$ for any $v\in V(G)\setminus \hset{v_1,v_2}$.
	\end{itemize}
	Intuitively, $v_2$ is merged into $v_1$ together with its weight.
	Then, we have the following lemma.
	\begin{lemma}
	\label{lem:wreduce}
	$\OPT{\Inst}{\weight}=\OPT{\Inst^\prime}{\weight^\prime}$.
	\end{lemma}
	\begin{proof}
		For the notational convenience, we denote $G^\prime:=G\setminus H_2$.
		By Lemma~\ref{lem:reduce}, $\OPT{\Inst}{\weight}=+\infty$ if and only if $\OPT{\Inst^\prime}{\weight^\prime}=+\infty$.
		Therefore, we assume that $\OPT{\Inst^\prime}{\weight^\prime}\neq +\infty$ and $\OPT{\Inst}{\weight}\neq +\infty$.
		
		We first show that $\OPT{\Inst}{\weight}\le \OPT{\Inst^\prime}{\weight^\prime}$.
		Since $\OPT{\Inst}{\weight}\le \len{\weight}{\Seq}$ holds for any reconfiguration sequence $\Seq$ for $\Inst$, it suffices to show that there exists a reconfiguration sequence for $\Inst$ whose length is at most $\OPT{\Inst^\prime}{\weight^\prime}$.
		Let $\Seq^\prime$ be a shortest reconfiguration sequence for $\Inst^\prime$ such that $\len{\weight^\prime}{\Seq^\prime}=\OPT{\Inst^\prime}{\weight^\prime}$.
		Following the only-if direction proof of Lemma~\ref{lem:reduce}, we can construct a reconfiguration sequence $\Seq$ for $\Inst$ such that $\nrec{\Seq}{v_1}=\nrec{\Seq}{v_2}=\nrec{\Seq^\prime}{v_1}$ and $\nrec{\Seq}{v}=\nrec{\Seq^\prime}{v}$ for any $v\in V(G)\setminus \hset{v_1,v_2}$.
		Therefore, 
		\[
		\begin{array}{lll}
		\len{\weight}{\Seq} &=& \sum_{v\in V(G)}\weight(v)\cdot \nrec{\Seq}{v} \\
		&=& \weight(v_1) \cdot \nrec{\Seq}{v_1} + \weight(v_2) \cdot \nrec{\Seq}{v_2} + \sum_{v\in V(G)\setminus \hset{v_1,v_2}}\weight(v)\cdot \nrec{\Seq}{v} \\
		&=& (\weight(v_1) + \weight(v_2))\cdot \nrec{\Seq}{v_1} + \sum_{v\in V(G)\setminus \hset{v_1,v_2}}\weight(v)\cdot \nrec{\Seq}{v} \\
		&=& \weight^\prime(v_1) \cdot \nrec{\Seq^\prime}{v_1} + \sum_{v\in V(G)\setminus \hset{v_1,v_2}}\weight^\prime(v)\cdot \nrec{\Seq^\prime}{v} \\
		&=& \sum_{v\in V(G^\prime)}\weight^\prime(v)\cdot \nrec{\Seq^\prime}{v} \\
		&=& \len{\weight^\prime}{\Seq^\prime} \\
		&=& \OPT{\Inst^\prime}{\weight^\prime}.
		\end{array}
		\]
		Thus, $\Seq$ is a desired reconfiguration sequence for $\Inst$.
		
		We next show that $\OPT{\Inst^\prime}{\weight^\prime}\le \OPT{\Inst}{\weight}$.
		Since $\OPT{\Inst^\prime}{\weight^\prime}\le \len{\weight^\prime}{\Seq^\prime}$ holds for any reconfiguration sequence $\Seq^\prime$ for $\Inst^\prime$, it suffices to show that there exists a reconfiguration sequence for $\Inst^\prime$ whose length is at most $\OPT{\Inst}{\weight}$.
		Let $\Seq$ be a shortest reconfiguration sequence for $\Inst$ such that $\len{\weight}{\Seq}=\OPT{\Inst}{\weight}$.
		We now construct a reconfiguration sequence for $\Inst^\prime$ from $\Seq$ such that $\len{\weight^\prime}{\Seq^\prime}\le \OPT{\Inst}{\weight}$ as follows.
		\begin{description}
			\item[Case 1. $\nrec{\Seq}{v_1}\le \nrec{\Seq}{v_2}$:]
			In this case, we restrict all $\verlist$-colorings in $\Seq$ on $V(G^\prime)$ to obtain a reconfiguration sequence $\Seq_1$ for $\rest{\Inst}{G^\prime}=\Inst^\prime$; recall the if direction proof of Lemma~\ref{lem:reduce}.
			From the construction, $\nrec{\Seq_1}{v_1}=\nrec{\Seq}{v_1}\le \nrec{\Seq}{v_2}$ and $\nrec{\Seq_1}{v}=\nrec{\Seq}{v}$ holds for any vertex $v\in V(G^\prime)=V(G) \setminus \hset{v_2}$.
			Therefore, we have
			\[
			\begin{array}{lll}
			\len{\weight^\prime}{\Seq_1} &=& \sum_{v\in V(G^\prime)}\weight^\prime(v)\cdot \nrec{\Seq_1}{v} \\
			&=& \weight^\prime(v_1) \cdot \nrec{\Seq_1}{v_1} + \sum_{v\in V(G^\prime)\setminus \hset{v_1}}\weight^\prime(v)\cdot \nrec{\Seq_1}{v} \\
			&=& (\weight(v_1)+\weight(v_2)) \cdot \nrec{\Seq}{v_1} + \sum_{v\in V(G)\setminus \hset{v_1,v_2}}\weight(v)\cdot \nrec{\Seq}{v} \\
			&\le & \weight(v_1) \cdot \nrec{\Seq}{v_1} + \weight(v_2) \cdot \nrec{\Seq}{v_2} + \sum_{v\in V(G)\setminus \hset{v_1,v_2}}\weight(v)\cdot \nrec{\Seq}{v} \\
			&=& \sum_{v\in V(G)}\weight(v)\cdot \nrec{\Seq}{v} \\
			&=& \len{\weight}{\Seq} \\
			&=& \OPT{\Inst}{\weight}.
			\end{array}
			\]
			Thus, $\Seq_1$ is a desired reconfiguration sequence for $\Inst^\prime$.\\
			\item[Case 2. $\nrec{\Seq}{v_1} > \nrec{\Seq}{v_2}$:]
			In this case, instead of restricting $\verlist$-colorings in $\Seq$ on $V(G^\prime)$, we restrict them on $V(G\setminus H_1)$ and obtain a reconfiguration sequence $\Seq_2$ for $\rest{\Inst}{G\setminus H_1}$.
			Then, because $H_1$ and $H_2$ are identical, we can easily ``rephrase'' $\Seq_2$ as a reconfiguration sequence $\Seq_2^\prime$ for $\Inst^\prime$.
			By the same arguments as the case 1 above, we have $\len{\weight^\prime}{\Seq_2^\prime} < \len{\weight}{\Seq}= \OPT{\Inst}{\weight}$.
			Thus, $\Seq_2^\prime$ is a desired reconfiguration sequence for $\Inst^\prime$.
		\end{description}
		
		In this way, we have shown that $\OPT{\Inst}{\weight}=\OPT{\Inst^\prime}{\weight^\prime}$ as claimed.
		\qed
	\end{proof}


	
\subsection{Kernelazation}

	Finally, we give a kernelization algorithm as follows.
	
	Let $(\Inst=(G,\verlist,\mapf_{\ini},\mapf_{\tar}),\weight)$ be an instance of the weighted version such that $G$ has a vertex cover of size at most $\VC$.
	Because such a vertex cover can be computed in time $O(2^\VC \cdot |V(G)|)$~\cite{DF99}, we now assume that we are given a vertex cover $\Vcov$ of size at most $\VC$.
	Notice that $\Vind:=V\setminus \Vcov$ forms an independent set of $G$.
	Suppose that there exist two vertices $v_1,v_2\in \Vind$ such that $N(G,v_1)=N(G,v_2)$ and $\asgn{v_1}=\asgn{v_2}$ hold.
	Then, induced subgraphs $G[\hset{v_1}]$ and $G[\hset{v_2}]$ are identical.
	Therefore, we can apply Lemma~\ref{lem:wreduce} to remove $v_2$ from $G$, and modify a weight function without changing the optimality.
	As a kernelization, we repeatedly apply Lemma~\ref{lem:wreduce} for all such pairs of vertices in $\Vind$, which can be done in polynomial time.
	Let $G^\prime$ be the resulting subgraph of $G$, and let $\Vind^\prime:=V(G^\prime)\setminus \Vcov$.
	Since $\Vcov$ is of size at most $\VC$, it suffices to prove the following lemma.
	\begin{lemma}
	\label{lem:kernel_sp}
	$|\Vind^\prime|\le 2^\VC \cdot 2^k\cdot k^2$.
	\end{lemma}
	\begin{proof}
		Recall that $\Vind^\prime$ contains no pair of vertices which induce identical subgraphs, and hence any pair of vertices $v_1,v_2\in \Vind^\prime$ does not satisfy at least one of $N(G,v_1)=N(G,v_2)$ and $\asgn{v_1}=\asgn{v_2}$.
		Therefore, $|\Vind^\prime|$ can be bounded by the number of distinct combinations of the neighborhood and the vertex assignment.
		Since $\Vind^\prime$ is an independent set, $\Neigh{G^\prime}{v}\subseteq \Vcov^\prime$ for each vertex $v\in \Vind^\prime$.
		Recall that $|\Vcov^\prime|\le \VC$, and hence the number of (possible) neighborhoods can be bounded by $2^\VC$.
		Since there are at most $k$ colors, the number of (possible) vertex assignments can be bounded by $2^k\cdot k^2$.
		We thus have $|\Vind^\prime|\le 2^\VC \cdot 2^k\cdot k^2$ as claimed.
		\qed
	\end{proof}
	
	This completes the proof of Theorem~\ref{the:split}.
	\section{W[1]-hardness}
\label{sec:hardness}

	Because even the shortest  variant is fixed-parameter tractable when parameterized by $k+\VC$, one may expect that $\VC$ is a strong parameter and the problem is fixed-parameter tractable with only $\VC$.
	However, we prove the following theorem in this section.
	\begin{theorem}
		\label{the:W1}
		\textsc{List coloring reconfiguration} is $W[1]$-hard when parameterized by $\VC$, where $\VC$ is the upper bound on the size of a minimum vertex cover of an input graph.
	\end{theorem}
	Recall that \textsc{list coloring reconfiguration} is PSPACE-complete even for a fixed constant $k \ge 4$.
	Therefore, the problem is intractable if we take only one parameter, either $k$ or $\VC$. 
	
	\medskip
	In order to prove Theorem~\ref{the:W1}, we give an FPT-reduction from the \textsc{independent set} problem when parameterized by the solution size $s$, in which we are given a graph $H$ and an integer $s\ge 0$, and asked whether $H$ has an independent set of size at least $s$.
	This problem is known to be W[1]-hard~\cite{DF99}.
	
\subsection{Construction}

	Let $H$ be a graph with $n$ vertices $u_1,u_2,\ldots,u_n$, and $s$ be an integer as an input for \textsc{independent set}.
	Then, we construct the corresponding instance $(G,\verlist,\mapf_{\ini},\mapf_{\tar})$ of \textsc{list coloring reconfiguration} as follows.
	(See also \figurename~\ref{fig:reduction}.)

\begin{figure}[t]
	\begin{center}
		\includegraphics[width=0.8\textwidth,clip]{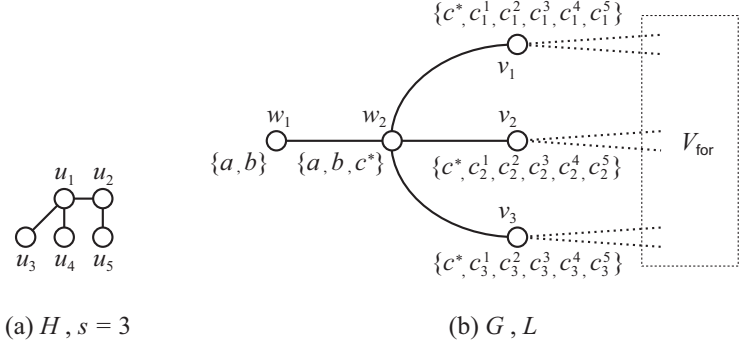}
	\end{center}
	\caption{(a) An instance $H$ of \textsc{independent set}, and (b) the graph $G$ and the list $L$.
		The set $\Vfbd$ contains vertices of $(i,j;p,q)$-forbidding gadgets for all $(i,j)\in \hset{(1,2),(1,3),(2,3)}$ 
		and all $(p,q)\in \hset{(1,1),(2,2),(3,3),(4,4),(5,5),\allowbreak (1,2),(2,1),(1,3),(3,1),(1,4),(4,1),(2,5),(5,2)}$;
		thus $|\Vfbd|=39$.}
	\label{fig:reduction}
\end{figure}

	We first create $s$ vertices $v_1, v_2, \ldots, v_s$, which are called \emph{selection vertices}; let $\select$ be the set of all selection vertices.
	For each $i\in \hset{1,2,\ldots,s}$, we set $\verlist(v_i)=\hset{\key,c_i^1,c_i^2,\ldots,c_i^n}$.
	In our reduction, we will construct $G$ and $L$ so that assigning the color $c_i^p$, $p \in \hset{1,2,\ldots,n}$, to $v_i \in \select$ corresponds to choosing the vertex $u_p \in V(H)$ as a vertex in an independent set of $H$.
	Then, in order to make a correspondence between a color assignment to $\select$ and an independent set of size $s$ in $H$, we need to construct the following properties:
	\begin{itemize}
		\item For each $p \in \{1,2,\ldots,n\}$, we use at most one color from $\hset{c_1^p, c_2^p, \ldots, c_s^p}$; this ensures that each vertex $u_p \in V(H)$ can be chosen at most once as an independent set.
		\item For each $p, q \in \{1,2,\ldots,n\}$ with $u_p u_q \in E(H)$, we use at most one color from $\hset{c_1^p, c_2^p, \ldots, c_s^p,\allowbreak c_1^q, c_2^q, \ldots, c_s^q}$; then, no two adjacent vertices in $H$ are chosen as an independent set.
	\end{itemize}
	To do this, we define an $(i,j;p,q)$-forbidding gadget for $i,j\in \hset{1,2,\ldots s}$ and $p,q\in \hset{1,2,\ldots,n}$.
	The $(i,j;p,q)$-\emph{forbidding gadget} is a vertex $\fbd{i}{j}$ which is adjacent to $v_i$ and $v_j$ and has a list $\verlist(\fbd{i}{j})=\hset{c_i^q,c_j^p}$.
	Observe that the vertex $\fbd{i}{j}$ forbids that $v_i$ and $v_j$ are simultaneously colored with $c_i^p$ and $c_j^q$, respectively.
	In order to satisfy the desired properties above, we now add our gadgets as follows: for all $i,j\in \hset{1,2,\ldots s}$ with $i<j$,
	\begin{itemize}
		\item add an $(i,j;p,p)$-forbidding gadget for every vertex $u_p\in V(H)$; and
		\item add $(i,j;p,q)$- and $(i,j,q,p)$-forbidding gadgets for every edge $u_p u_q\in E(H)$.
	\end{itemize}
	We denote by $\Vfbd$ the set of all vertices in the forbidding gadgets.
	We finally create an edge consisting of two vertices $w_1$ and $w_2$ such that $\verlist(w_1)=\hset{a,b}$ and $\verlist(w_2)=\hset{a,b,\key}$, and connect $w_2$ with all selection vertices in $\select$.

	Finally, we construct two $\verlist$-colorings $\mapf_{\ini}$ and $\mapf_{\tar}$ of $G$ as follows:
	\begin{itemize}
		\item for each $v_i\in \select$, $\mapf_{\ini}(v_i)=\mapf_{\tar}(v_i)=\key$;
		\item for each $\fbd{i}{j}\in \Vfbd$, $\mapf_{\ini}(\fbd{i}{j})$ and $\mapf_{\tar}(\fbd{i}{j})$ are arbitrary chosen colors from $\verlist(\fbd{i}{j})$; and
		\item $\mapf_{\ini}(w_1)=\mapf_{\tar}(w_2)=a$, and $\mapf_{\tar}(w_1)=\mapf_{\ini}(w_2)=b$.
	\end{itemize}
	Note that both $\mapf_{\ini}$ and $\mapf_{\tar}$ are proper $\verlist$-colorings of $G$.
	
	In this way, we complete the construction of $(G,\verlist,\mapf_{\ini},\mapf_{\tar})$.
	
\subsection{Correctness of the reduction}

	In this subsection, we prove the following three statements:
	\begin{itemize}
		\item $(G,\verlist,\mapf_{\ini},\mapf_{\tar})$ can be constructed in time polynomial in the size of $H$.
		\item The upper bound $\VC$ on the size of a minimum vertex cover of $G$ depends only on $s$.
		\item $H$ is a yes-instance of \textsc{independent set} if and only if $(G,\verlist,\mapf_{\ini},\mapf_{\tar})$ is a yes-instance of \lcr.
	\end{itemize}

	In order to prove the first statement, it suffices to show that the size of $(G,\verlist,\mapf_{\ini},\mapf_{\tar})$ is bounded polynomially in $n=|V(H)|$.
	From the construction, we have $|V(G)|=|\select|+|\Vfbd|+|\hset{w_1,w_2}|\le s+s^2\times (|V(H)|+2|E(H)|) +2=O(n^4)$.
	In addition, each list contains $O(n)$ colors.
	Therefore, the construction can be done in time $O(n^{O(1)})$.
	
	The second statement immediately follows from the fact that $\hset{w_2}\cup \select$ is a vertex cover in $G$ of size $s+1$; observe that $G\setminus V^\prime=G[\hset{w_1}\cup \Vfbd]$ contains no edge.
	
	Finally, we prove the last statement as follows.
	\begin{lemma}
		\label{lem:correct}
		$H$ is a yes-instance of \textsc{independent set} if and only if $(G,\verlist,\mapf_{\ini},\mapf_{\tar})$ is a yes-instance of \lcr.
	\end{lemma}
	\begin{proof}
		We first prove the if direction.
		Assume that there exists a reconfiguration sequence $\Seq$ for $(G,\verlist,\mapf_{\ini},\mapf_{\tar})$.
		Then, $\Seq$ must contain at least one $\verlist$-coloring $\mapf$ such that $\mapf(w_2)=\key$ in order to recolor $w_1$ from $a$ to $b$.
		Since $w_2$ is adjacent to all vertices in $\select$, $\mapf(v_i)\ne \key$ holds for every $v_i\in \select$.
		Then, by the construction, the vertex set $\hset{ u_p \colon c_i^p = f(v_i), v_i \in \select }$ is an independent set in $H$ of size $|\select| = s$. 
		
		We then prove the only-if direction.
		We construct a reconfiguration sequence for $(G,\verlist,\mapf_{\ini},\mapf_{\tar})$ which passes through two $\verlist$-colorings $\mapf_{\ini}^\prime$ and $\mapf_{\tar}^\prime$ defined as follows.
		
		From the assumption, $H$ has an independent set $I$ of size $s$, say, $I=\hset{u_1,u_2,\ldots,u_s}$.
		Then, we define $\mapf_{\ini}^\prime$ as follows:
		\begin{itemize}
			\item for each $v_i\in \select$, $\mapf_{\ini}^\prime(v_i)=c_i^i$;
			\item for each $(i,j;p,q)$-forbidding vertex $\fbd{i}{j}\in \Vfbd$, $\mapf_{\ini}^\prime(\fbd{i}{j})$ is an arbitrary chosen color from $\verlist(\fbd{i}{j})\setminus \hset{c_i^i,c_j^j}$; and
			\item $\mapf_{\ini}^\prime(w_1)=\mapf_{\ini}(w_1)=a$ and $\mapf_{\ini}^\prime(w_2)=\mapf_{\ini}(w_2)=b$.
		\end{itemize}
		Note that $\mapf_{\ini}^\prime$ is a proper $\verlist$-coloring of $G$.
		We next show that $\mapf_{\ini}$ and $\mapf_{\ini}^\prime$ are reconfigurable.
		We first recolor all vertice $\fbd{i}{j}\in \Vfbd$ to the colors $\mapf_{\ini}^\prime(w)$ $(\ne \key)$ in an arbitrary order.
		This can be done, since $\mapf_{\ini}(v_i)=\key$ for all $v_i\in \select$ and $\Vfbd$ is an independent set in $G$.
		We then recolor all vertices $v_i\in \select$ to the colors $\mapf_{\ini}^\prime(v_i)$ in an arbitrary order.
		This also can be done, since $\mapf_{\ini}^\prime$ is a proper $\verlist$-coloring and $\select$ is an independent set in $G$.
		Thus, $\mapf_{\ini}$ and $\mapf_{\ini}^\prime$ are reconfigurable.
		
		By the similar arguments as $\mapf_{\ini}$, we define $\mapf_{\tar}^\prime$ as follows:
		\begin{itemize}
			\item for each $v_i\in \select$, $\mapf_{\tar}^\prime(v_i)=c_i^i$;
			\item for each $(i,j;p,q)$-forbidding vertex $\fbd{i}{j}\in \Vfbd$, $\mapf_{\tar}^\prime(\fbd{i}{j})$ is an arbitrary chosen color from $\verlist(\fbd{i}{j})\setminus \hset{c_i^i,c_j^j}$; and
			\item $\mapf_{\tar}^\prime(w_1)=\mapf_{\tar}(w_1)=b$ and $\mapf_{\tar}^\prime(w_2)=\mapf_{\tar}(w_2)=a$.
		\end{itemize}
		Then, $\mapf_{\tar}$ and $\mapf_{\tar}^\prime$ are reconfigurable.
		
		Finally, we prove that $\mapf_{\ini}^\prime$ and $\mapf_{\tar}^\prime$ are reconfigurable.
		Recall that $\mapf_{\ini}^\prime(w_1)=\mapf_{\tar}^\prime(w_2)=a$, $\mapf_{\tar}^\prime(w_1)=\mapf_{\ini}^\prime(w_2)=b$, and $\mapf_{\ini}^\prime(v_i)=\mapf_{\tar}^\prime(v_i)\ne \key$ for all $v_i\in \select$.
		Then, we can swap the color $a$ and $b$ by the following three steps:
		\begin{itemize}
			\item recolor $w_2$ to $\key$;
			\item recolor $w_1$ to $b$; and
			\item recolor $w_2$ to $a$.
		\end{itemize}
		After that, we can recolor all vertices $\fbd{i}{j}\in \Vfbd$ to the colors $\mapf_{\tar}^\prime(\fbd{i}{j})$ in the arbitrary order, since $\Vfbd$ is an independent in $G.$
		
		Therefore, $(G,\verlist,\mapf_{\ini},\mapf_{\tar})$ is a yes-instance of \lcr.
		\qed
	\end{proof}

	This completes the proof of Theorem~\ref{the:W1}.

	\section{Conclusion}
	
	In this paper, we have studied \lcr~from the viewpoint of parametrized complexity, in particular, with several graph parameters, 
	and painted an interesting map of graph parameters in \figurename~\ref{fig:result} which shows the boundary between fixed-parameter tractability and intractability. 
	

\bibliographystyle{abbrv}

\end{document}